\documentclass[11pt]{article}
\usepackage{framed,latexsym,amssymb,amsmath,amsthm,enumerate,geometry,float,cite,tikz,pgfplots,complexity}
\usepackage{verbatim}
\usepackage[shortend, vlined]{algorithm2e}
\usepackage{lineno}
\usepackage{setspace,xcolor}
\usepackage[shortlabels]{enumitem}

\allowdisplaybreaks
\pgfplotsset{compat=1.18}

\newtheorem{theorem}{Theorem}[section]
\newtheorem{proposition}[theorem]{Proposition}

\newtheorem{lemma}[theorem]{Lemma}

\newtheorem{corollary}[theorem]{Corollary}
\newtheorem{claim}{Claim}
\newtheorem{reductionrule}{Reduction Rule}[section]
\newtheorem{decisionrule}{Decision Rule}[section]
\newtheorem{problemenv}{Problem}[section]

\newcommand{\cfc}{\textsc{Conflict-Free Cut}}
\newcommand{\symnontrivialsat}{\textsc{Symmetric Non-Constant SAT}}

\newcommand{\calO}{\mathcal{O}}

\renewcommand{\SAT}{\textsc{3-SAT}}
\newcommand{\cleanSAT}{\textsc{Clean 3-SAT}}
\newcommand{\calF}{\mathcal{F}}
\newcommand{\matchingcut}{\textsc{Matching Cut}}

\begin{document}
\onehalfspace
\title{On Conflict-Free Cuts: Algorithms and Complexity}
\author{Johannes Rauch$^1$ \and Dieter Rautenbach$^1$ \and U\'{e}verton S. Souza$^2$}
\date{}
\maketitle
\vspace{-10mm}
\begin{center}
	{\small 
		$^1$ Institute of Optimization and Operations Research, Ulm University, Ulm, Germany\\
		$^2$ Instituto de Computa\c{c}\~{a}o, Universidade Federal Fluminense, Niter\'{o}i, Brazil\\
		\texttt{$\{$johannes.rauch,dieter.rautenbach$\}$@uni-ulm.de},
		\texttt{ueverton@ic.uff.br}
	}
\end{center}

\begin{abstract}
One way to define the \matchingcut{} problem is: Given a graph $G$, is there an edge-cut $M$ of $G$ such that $M$ is an independent set in the line graph of $G$? 
We propose the more general \cfc{} problem:
Together with the graph $G$, we are given a so-called conflict graph $\hat{G}$ on the edges of $G$, and we ask for an edge-cutset $M$ of $G$ that is independent in $\hat{G}$.
Since conflict-free settings are popular generalizations of classical optimization problems and \cfc{} was not considered in the literature so far, we start the study of the problem.
We show that the problem is \NP-complete even when the maximum degree of $G$ is 5 and $\hat{G}$ is 1-regular.
The same reduction implies an exponential lower bound on the solvability based on the Exponential Time Hypothesis.
We also give parameterized complexity results: We show that the problem is fixed-parameter tractable with the vertex cover number of $G$ as a parameter, and we show $\W[1]$-hardness even when $G$ has a feedback vertex set of size one, and the clique cover number of $\hat{G}$ is the parameter.
Since the clique cover number of $\hat{G}$ is an upper bound on the independence number of $\hat{G}$ and thus the solution size, this implies \W[1]-hardness when parameterized by the cut size.
We list polynomial-time solvable cases and interesting open problems.
At last, we draw a connection to a symmetric variant of \textsc{SAT}.
\end{abstract}

\section{Introduction}
The well-studied \matchingcut{} problem \cite{bonsma2012extremal,chen2021matching,chvatal1984recognizing,gomes2021finding,kratsch2016algorithms} asks whether a given graph $G$ admits a matching cut $M$, that is, a set of edges $M \subseteq E(G)$ such that $G-M$ is disconnected and $M$ is an independent set in the line graph of $G$.
Here, we generalize this problem to the \cfc{} problem.
Instead of the cut $M$ being an independent set in the line graph of $G$, we require that $M$ is an independent set in a so-called \emph{conflict graph $\hat{G}$} on the edges of $G$, which is given together with $G$ as an input.
Since \matchingcut{} was proven to be \NP-complete by Chv\'{a}tal in 1984~\cite{chvatal1984recognizing}, it is clear that \cfc{} is \NP-complete in the general case, too.
The conflict-free setting is a typical extension: Several classical optimization problems have been studied in a ``conflict-free'' version; among them are \textsc{Bin Packing}~\cite{capua2018study, gendreau2004heuristics}, \textsc{Knapsack}~\cite{pferschy2009knapsack,bettinelli2017branch}, \textsc{Maximum Matching} and \textsc{Shortest Path}~\cite{darmann2011paths}, \textsc{Hypergraph Matching}~\cite{glock2023conflict}, \textsc{Minimum Spanning Tree}~\cite{darmann2011paths,zhang2011minimum,barros2022conflict}, and \textsc{Set Cover}~\cite{jacob2021parameterized}.
However, as far as we know, \cfc{} is a new and unexplored problem.

All graphs in this article are finite, simple, and undirected, unless stated otherwise. We use standard terminology.
An \emph{(edge-)cut} of a graph $G$ is a set of edges $F$ such that $G-F$ is disconnected.
In this article, a cut will always be an edge-cut.
If $s$ and $t$ are two vertices of $G$, and $s$ and $t$ lie in different components of $G-F$, then $F$ is an \emph{$s$-$t$ cut}.
Formally, \cfc{} is the following problem:
\begin{framed}
\noindent
\cfc{} \\
Instance: A connected graph $G$, and a conflict graph $\hat{G}$ on $E(G)$, that is, $V(\hat{G}) = E(G)$. \\
Question: Does $G$ have a \emph{conflict-free cut}? That is, is there a subset $F \subseteq E(G)$ such that $G-F$ is disconnected and $F$ is independent in $\hat{G}$?
\end{framed}

Although \cfc{} is \NP-complete when the conflict graph is a line graph, as it generalizes \textsc{Matching Cut}, it is natural to ask about the complexity of \cfc{} when the conflict graph has a simpler structure. 
In particular, one of the simplest structures we can consider is the conflict graph being 1-regular. 
The complexity of conflict-free optimization problems when the conflict graph is 1-regular is intriguing. While \textsc{Maximum Matching} becomes NP-hard, \textsc{Minimum Spanning Tree} remains polynomial-time solvable under this constraint setting (see~\cite{darmann2011paths}). Therefore, it is interesting to investigate the complexity of finding a conflict-free cut under the same constraint setting.

The paper is structured as follows.
In Section~\ref{sec_np}, we prove that \cfc{} stays \NP{}-complete, even when $G$ has maximum degree at most five and the conflict graph $\hat{G}$ is 1-regular.
We show this by a linear reduction from a \SAT{} variant, which also implies a lower bound based on the Exponential Time Hypothesis (ETH) by Impagliazzo and Paturo~\cite{impagliazzo2001complexity}.
In Section~\ref{sec_param}, we give some parameterized complexity results. 
In particular, we show that \cfc{} is fixed-parameter tractable, taking the vertex cover number of $G$ as a parameter. 
In addition, we show that determining whether an instance $(G,\hat{G})$ admits a conflict-free cut is \W[1]-hard, even when $G$ is a series-parallel graph having a feedback vertex set of size one, and the clique cover number of $\hat{G}$ is the parameter. 
Note that the clique cover number of $\hat{G}$ upper bounds the independence number of $\hat{G}$, and thus the size of any conflict-free cut of $(G,\hat{G})$ too, implying \W[1]-hardness taking the cut size as a parameter as well.
Besides, we remark on how Courcelle's Theorem~\cite{courcelle1990graph} applies to \cfc{}. 
Section~\ref{sec_poly_open} lists some simple polynomial-time solvable cases and states some open problems regarding \cfc{}.
At last, in Section~\ref{sec_sat}, we introduce a variant of \textsc{SAT} that is closely related to \cfc{}, and prove that the problem reduces to solving $n$ many \textsc{4-SAT} instances, where $n$ is the order of the input graph $G$.

\subsection*{Preliminaries} \label{sec_pre}

We introduce some further terminology.
We exclusively refer to $G$ as the \emph{input graph} and to $\hat{G}$ as the \emph{conflict graph}.
When we say that two edges of a graph are \emph{conflicting}, this means that they are adjacent in the corresponding conflict graph.
At times, for the sake of simplicity, we refer to the conflict graph only implicitly by stating that a graph $G$ has edge conflicts, or writing that $G$ has a conflict-free cut.
For two disjoint sets of vertices $A$ and $B$ of a graph $G$, let $E(A,B)$ denote the set of edges of $G$ with one end in $A$ and the other in $B$. 
Besides, let $\partial_GA = E(A, V(G) \setminus A)$ denote the set of edges of $G$ with exactly one end in $A$.
For a vertex $v$ of $G$, we write 
$\partial_Gv$ instead of $\partial_G\{v\}$.
We will use the $\calO^*$-notation to suppress polynomial factors in the $\calO$-notation.


\section{An \NP-Completeness Proof} \label{sec_np}
We prove here that \cfc{} remains \NP-complete even when the input graph has maximum degree at most five and the conflict graph is $1$-regular. We begin by stating a result from the literature, on which our proof will rely.

Following~\cite{cygan2017hitting}, we say that a 3-CNF formula $\calF$ is \emph{clean} if each variable of $\calF$ appears exactly three times, at least once positively and at least once negatively, and each clause of $\calF$ contains two or three literals and does not contain one variable twice.
By way of renaming, we may additionally assume that each variable occurs twice positively.
We call the problem of deciding satisfiability of a clean 3-CNF formula \cleanSAT{}.
The reduction in the proof of Lemma~6 in \cite{cygan2017hitting} shows that \cleanSAT{} is \NP-complete.

Moreover, we need two types of gadgets.
Each gadget is a graph with edge conflicts. 
The first type is the \emph{uncutable gadget}. 
It is the square of a cycle $C_{2n}$ with even order, $n \geq 4$, and two additional edges.
The two additional edges have one end in each partite set when viewing $C_{2n}$ as a bipartite graph, respectively.
The edge conflicts of the uncutable gadget are as indicated in Figure~\ref{fig_gadgets}.
As the name indicates, these graphs do not admit a conflict-free cut.
Without the two additional edges, it is easy to see that the only conflict-free cut is separating the partite sets of the underlying $C_{2n}$, again viewing it as a bipartite graph.
With the two additional edges, this cut contains two conflicting edges, too.
Note that all uncutable gadgets have two vertices of degree five, while the rest of the vertices have degree four.
Note also that they are non-planar.
The second type is the \emph{variable gadget}. 
It is also depicted in Figure~\ref{fig_gadgets}. 
Note that the corresponding conflict graph has maximum degree one for both gadget types.

With this we are ready to prove Theorem~\ref{thm_np}.

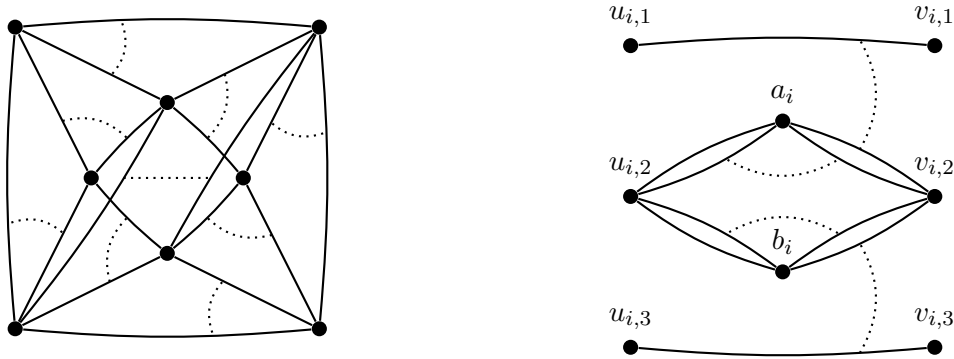
\begin{figure}
\begin{minipage}{0.5\textwidth}
\centering
\begin{tikzpicture}

\node [circle, fill,inner sep=2] (v6) at (-1,0) {};
\node [circle, fill,inner sep=2] (v8) at (1,0) {};
\node [circle, fill,inner sep=2] (v7) at (0,1) {};
\node [circle, fill,inner sep=2] (v5) at (0,-1) {};
\node [circle, fill,inner sep=2] (v1) at (-2,-2) {};
\node [circle, fill,inner sep=2] (v4) at (2,-2) {};
\node [circle, fill,inner sep=2] (v3) at (2,2) {};
\node [circle, fill,inner sep=2] (v2) at (-2,2) {};

\draw [thick](v5) -- (v1) -- (v6) -- (v2) -- (v7) -- (v3) -- (v8) -- (v4) -- (v5);

\node [inner sep=0] (v10) at (-1.35,-0.75) {};
\node [inner sep=0](v9) at (-2.075,-0.6) {};
\node [inner sep=0](v12) at (-0.75,1.35) {};
\node [inner sep=0](v11) at (-0.6,2.075) {};
\node [inner sep=0](v14) at (1.35,0.75) {};
\node [inner sep=0](v13) at (2.075,0.6) {};
\node [inner sep=0](v16) at (0.75,-1.35) {};
\node [inner sep=0](v15) at (0.6,-2.075) {};
\draw [thick, dotted, bend left] (v9) edge (v10);
\draw [thick, dotted, bend left] (v11) edge (v12);
\draw [thick, dotted, bend left] (v13) edge (v14);
\draw [thick, dotted, bend left] (v15) edge (v16);

\node [inner sep=0](v18) at (0.5,-0.5) {};
\node [inner sep=0](v17) at (1.4,-0.75) {};
\node [inner sep=0](v24) at (-0.5,-0.5) {};
\node [inner sep=0](v22) at (-0.5,0.5) {};
\node [inner sep=0](v20) at (0.5,0.5) {};
\node [inner sep=0](v23) at (-0.75,-1.4) {};
\node [inner sep=0](v21) at (-1.4,0.75) {};
\node [inner sep=0](v19) at (0.75,1.4) {};
\draw [thick, dotted, bend left] (v17) edge (v18);
\draw [thick, dotted, bend left] (v19) edge (v20);
\draw [thick, dotted, bend left] (v21) edge (v22);
\draw [thick, dotted, bend left] (v23) edge (v24);
\draw [thick, bend left=5] (v1) edge (v2);
\draw [thick, bend left=5] (v2) edge (v3);
\draw [thick, bend left=5] (v3) edge (v4);
\draw [thick, bend left=5] (v4) edge (v1);

\draw [thick, bend left=5] (v7) edge (v1);
\draw [thick, bend right=5] (v3) edge (v5);
\node [inner sep=0] (v25) at (-0.6,0) {};
\node [inner sep=0] (v26) at (0.6,0) {};
\draw [thick, dotted] (v25) edge (v26);
\draw [thick, bend left=5] (v6) edge (v7);
\draw [thick, bend left=5] (v7) edge (v8);
\draw [thick, bend left=5] (v8) edge (v5);
\draw [thick, bend left=5] (v5) edge (v6);
\end{tikzpicture}
\end{minipage}
\begin{minipage}{0.5\textwidth}
\centering
\begin{tikzpicture}

\node [fill, circle, inner sep=2, label=$u_{i,1}$] (v1) at (-2,2) {};
\node [fill, circle, inner sep=2, label=$v_{i,1}$] (v2) at (2,2) {};
\node [fill, circle, inner sep=2, label=$u_{i,2}$] (v3) at (-2,0) {};
\node [fill, circle, inner sep=2, label=$v_{i,2}$] (v6) at (2,0) {};
\node [fill, circle, inner sep=2, label=$u_{i,3}$] (v7) at (-2,-2) {};
\node [fill, circle, inner sep=2, label=$v_{i,3}$] (v8) at (2,-2) {};
\node [fill, circle, inner sep=2, label=$a_i$] (v4) at (0,1) {};
\node [fill, circle, inner sep=2, label=$b_i$] (v5) at (0,-1) {};
\draw [thick, bend left=5] (v1) edge (v2);
\draw [thick, bend left=10] (v3) edge (v4);
\draw [thick, bend right=10] (v3) edge (v5);
\draw [thick, bend left=10] (v4) edge (v6);
\draw [thick, bend right=10] (v5) edge (v6);
\draw [thick, bend right=5] (v7) edge (v8);
\node [inner sep=0] (v12) at (1,2.1) {};
\node [inner sep=0] (v11) at (1,0.6) {};
\node [inner sep=0] (v10) at (1,-0.6) {};
\node [inner sep=0] (v9) at (1,-2.1) {};
\draw [thick, dotted, bend right] (v9) edge (v10);
\draw [thick, dotted, bend right] (v11) edge (v12);

\draw [thick, bend left=10] (v3) edge (v5);
\draw [thick, bend left=10] (v5) edge (v6);
\node [inner sep=0] (v14) at (-0.75,-0.5) {};
\node [inner sep=0] (v13) at (0.75,-0.5) {};
\draw [thick, bend right=10] (v3) edge (v4);
\draw [thick, bend right=10] (v4) edge (v6);
\node [inner sep=0] (v15) at (-0.75,0.5) {};
\node [inner sep=0] (v16) at (0.75,0.5) {};
\draw [thick, dotted, bend right] (v13) edge (v14);
\draw [thick, dotted, bend right] (v15) edge (v16);
\end{tikzpicture}
\end{minipage}
\caption{On the left is the uncutable gadget on 8 vertices. On the right is the variable gadget. The edge conflicts are indicated as dotted lines.} \label{fig_gadgets}
\end{figure}

\begin{theorem} \label{thm_np}
\cfc{} is \NP{}-complete, even when the input graph has maximum degree at most five and the conflict graph is 1-regular.
\end{theorem}
\begin{proof}
\cfc{} clearly is in \NP{}.
It remains to prove \NP{}-hardness.
For this we reduce \cleanSAT{}, which is \NP-complete, to \cfc{}.
Let $\calF$ be a $n$-variable $m$-clause instance of \cleanSAT{},
and let $x_1, \dots, x_n$ denote the variables of $\calF$.
Note that $m = \calO(n)$.
The properties of the formula $\calF$ allow us to view it as a family of literal sets, which we will use at times.

We construct a graph $G$ with edge conflicts as an instance of \cfc{} as follows.
First, we construct a multigraph $G'$ with edge conflicts and two designated vertices $s$ and $t$.
The induced conflict graph will be 1-regular.
If $G'$ has a conflict-free cut, then it will be a conflict-free $s$-$t$ cut.
Moreover, $\calF$ is satisfiable if and only if $G'$ has a conflict-free cut.
After that, we will show how to obtain a (simple) graph $G$ with maximum degree five from the multigraph $G'$.
We will construct $G$ such that it has a conflict-free cut if and only if $G'$ has one. 
Since the described construction is possible in polynomial time, the reduction is then complete.

We begin with the construction of $G'$.
Let $s$ and $t$ be two distinct vertices, which we add to $G'$.
For every clause $C \in \calF$ we add an $s$-$t$ path $P_C$ of length $|C|$ to $G'$.
The paths $P_C$ are internally (vertex-)disjoint.
For every clause $C \in \calF$, and every pair of distinct edges $e,f \in E(P_C)$, we create a copy of $e$ and a copy of $f$ to obtain two new edges $e'$ and $f'$, respectively, and add an edge conflict between $e'$ and $f'$.
We call the newly created copies of $e \in E(P_C)$ the \emph{shadow edges of $e$}, and denote the multiset of $e$ together with its shadow edges by $S(e)$. We refer to the union of all shadow edges plainly as the \emph{shadow edges}.
At this point, $G'$ looks like the multigraph in Figure~\ref{fig_step1}.

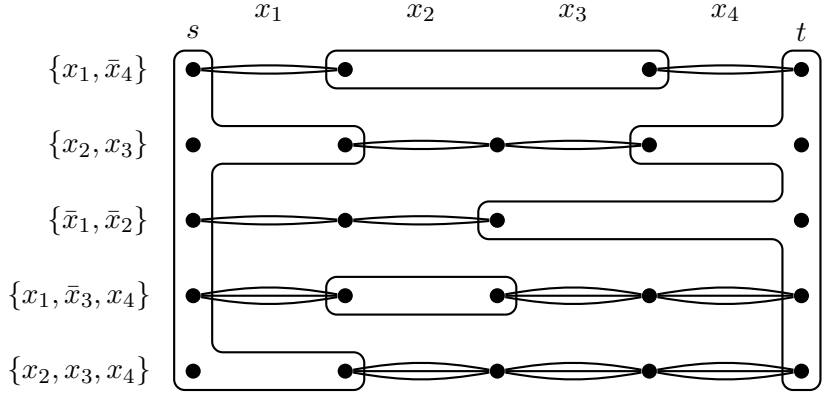
\begin{figure}
\centering
\begin{tikzpicture}

\node [circle, fill, inner sep=2] (l1) at (-4,2) {};
\node [circle, fill, inner sep=2] (l2) at (-4,1) {};
\node [circle, fill, inner sep=2] (l3) at (-4,0) {};
\node [circle, fill, inner sep=2] (l4) at (-4,-1) {};
\node [circle, fill, inner sep=2] (l5) at (-4,-2) {};

\node [circle, fill, inner sep=2] (u1) at (-2,2) {};
\node [circle, fill, inner sep=2] (u2) at (-2,1) {};
\node [circle, fill, inner sep=2] (u3) at (-2,0) {};
\node [circle, fill, inner sep=2] (u4) at (-2,-1) {};
\node [circle, fill, inner sep=2] (u5) at (-2,-2) {};

\node [circle, fill, inner sep=2] (v2) at (0,1) {};
\node [circle, fill, inner sep=2] (v3) at (0,0) {};
\node [circle, fill, inner sep=2] (v4) at (0,-1) {};
\node [circle, fill, inner sep=2] (v5) at (0,-2) {};

\node [circle, fill, inner sep=2] (w1) at (2,2) {};
\node [circle, fill, inner sep=2] (w2) at (2,1) {};
\node [circle, fill, inner sep=2] (w4) at (2,-1) {};
\node [circle, fill, inner sep=2] (w5) at (2,-2) {};

\node [circle, fill, inner sep=2] (r1) at (4,2) {};
\node [circle, fill, inner sep=2] (r2) at (4,1) {};
\node [circle, fill, inner sep=2] (r3) at (4,0) {};
\node [circle, fill, inner sep=2] (r4) at (4,-1) {};
\node [circle, fill, inner sep=2] (r5) at (4,-2) {};

\draw [thick, rounded corners] (-4.25,0) -- (-4.25,2.25) -- (-3.75,2.25) -- (-3.75,1.25) -- (-1.75,1.25) -- (-1.75,0.75) -- (-3.75,0.75) -- (-3.75,-1.75) -- (-1.75,-1.75) -- (-1.75,-2.25) -- (-3.75,-2.25) -- (-4.25,-2.25) -- (-4.25,0);
\draw [thick, rounded corners] (4.25,0) -- (4.25,2.25) -- (3.75,2.25) -- (3.75,1.25) -- (1.75,1.25) -- (1.75,0.75) -- (3.75,0.75) -- (3.75,0.25) -- (-0.25,0.25) -- (-0.25,-0.25) -- (3.75,-0.25) -- (3.75,-2.25) -- (4.25,-2.25) -- (4.25,0);

\draw [thick, rounded corners] (-2.25,2.25) rectangle (2.25,1.75);
\draw [thick, bend right=-5] (l1) edge (u1);
\draw [thick, bend right=5] (l1) edge (u1);
\draw [thick, bend right=-5] (w1) edge (r1);
\draw [thick, bend right=5] (w1) edge (r1);

\draw [thick, bend right=-5] (u2) edge (v2);
\draw [thick, bend right=5] (u2) edge (v2);
\draw [thick, bend right=-5] (v2) edge (w2);
\draw [thick, bend right=5] (v2) edge (w2);

\draw [thick, bend right=-5] (l3) edge (u3);
\draw [thick, bend right=5] (l3) edge (u3);
\draw [thick, bend right=-5] (u3) edge (v3);
\draw [thick, bend right=5] (u3) edge (v3);

    \draw [thick, rounded corners] (-2.25,-1.25) rectangle (0.25,-0.75);
    \draw [thick, bend right=-10] (l4) edge (u4);
\draw [thick, bend right=0] (l4) edge (u4);
\draw [thick, bend right=10] (l4) edge (u4);
\draw [thick, bend right=-10] (v4) edge (w4);
\draw [thick, bend right=0] (v4) edge (w4);
\draw [thick, bend right=10] (v4) edge (w4);
\draw [thick, bend right=-10] (w4) edge (r4);
\draw [thick, bend right=0] (w4) edge (r4);
\draw [thick, bend right=10] (w4) edge (r4);

\draw [thick, bend right=-10] (u5) edge (v5);
\draw [thick, bend right=0] (u5) edge (v5);
\draw [thick, bend right=10] (u5) edge (v5);
\draw [thick, bend right=-10] (v5) edge (w5);
\draw [thick, bend right=0] (v5) edge (w5);
\draw [thick, bend right=10] (v5) edge (w5);
\draw [thick, bend right=-10] (w5) edge (r5);
\draw [thick, bend right=0] (w5) edge (r5);
\draw [thick, bend right=10] (w5) edge (r5);

\node at (-4,2.5) {$s$};
\node at (4,2.5) {$t$};
\node at (-3,2.75) {$x_1$};
\node at (-1,2.75) {$x_2$};
\node at (1,2.75) {$x_3$};
\node at (3,2.75) {$x_4$};

\node at (-5.25,2) {$\{ x_1, \bar{x}_4 \}$};
\node at (-5.25,1) {$\{ x_2, x_3 \}$};
\node at (-5.25,0) {$\{ \bar{x}_1, \bar{x}_2 \}$};
\node at (-5.5,-1) {$\{ x_1, \bar{x}_3, x_4 \}$};
\node at (-5.5,-2) {$\{ x_2, x_3, x_4 \}$};
\end{tikzpicture}
\caption{The unfinished multigraph $G'$ right after adding the copied edges for the family of clauses $\calF = \left\{ \{x_1, \bar{x}_4\}, \{x_2, x_3\}, \{\bar{x}_1, \bar{x}_2\}, \{x_1, \bar{x}_3, x_4\}, \{x_2, x_3, x_4\} \right\}$. The vertices in a rectangle are thought of as identified, and the edge conflicts are not drawn.}
\label{fig_step1}
\end{figure}

If $x_{i_1}, \dots, x_{i_k}$, $i_1 < \dots < i_k$, are the variables of a clause $C \in \calF$, then we \emph{associate} the $j$-th edge of $P_C$ (as seen from $s$) to the variable $x_{i_j}$.
Now we modify the paths $P_C$ and add edge conflicts for all $i \in [n]$, thereby implicitly constructing some variable gadgets.
Let $C \in \calF$ be the clause where $x_i$ occurs negatively, and let $u_{i,2}v_{i,2}$ be the edge associated to $x_i$ in the path $P_C$.
(The indices in this paragraph refer to the variable gadget depicted in Figure~\ref{fig_gadgets} on the right.)
We remove $u_{i,2}v_{i,2}$ and add two paths $u_{i,2}a_iv_{i,2}$ and $u_{i,2}b_iv_{i,2}$ of length two between $u_{i,2}$ and $v_{i,2}$, where $a_i$ and $b_i$ are two new vertices.
For $u_{i,2}a_i$ we add a new $st$-edge conflicting with $u_{i,2}a_i$, and for $u_{i,2}b_i$ we add a new $st$-edge conflicting with $u_{i,2}b_i$, too.
We create a copy of $u_{i,2}a_i$ and a copy of $a_iv_{i,2}$, and add an edge conflict between the newly created copies.
We create a copy of $u_{i,2}b_i$ and $b_iv_{i,2}$, and add an edge conflict between the newly created copies.
Moreover, we add edge conflicts between $u_{i,1}v_{i,1}$ and $a_iv_{i,2}$, and between $u_{i,3}v_{i,3}$ and $b_iv_{i,2}$, where $u_{i,1}v_{i,1}$ and $u_{i,3}v_{i,3}$ are the remaining edges associated to (the positive occurrences of) $x_i$.
We have created a variable gadget for $x_i$ as in Figure~\ref{fig_gadgets}.

\begin{figure}
\centering
\begin{tikzpicture}
\node [circle, fill, inner sep=2] (l1) at (-4,2) {};
\node [circle, fill, inner sep=2] (l2) at (-4,1) {};
\node [circle, fill, inner sep=2] (l3) at (-4,0) {};
\node [circle, fill, inner sep=2] (l4) at (-4,-1) {};
\node [circle, fill, inner sep=2] (l5) at (-4,-2) {};

\node [circle, fill, inner sep=2] (lu3o) at (-3,0.375) {};
\node [circle, fill, inner sep=2] (lu3u) at (-3,-0.375) {};

\node [circle, fill, inner sep=2] (u1) at (-2,2) {};
\node [circle, fill, inner sep=2] (u2) at (-2,1) {};
\node [circle, fill, inner sep=2] (u3) at (-2,0) {};
\node [circle, fill, inner sep=2] (u4) at (-2,-1) {};
\node [circle, fill, inner sep=2] (u5) at (-2,-2) {};

\node [circle, fill, inner sep=2] (uv3o) at (-1,0.375) {};
\node [circle, fill, inner sep=2] (uv3u) at (-1,-0.375) {};

\node [circle, fill, inner sep=2] (v2) at (0,1) {};
\node [circle, fill, inner sep=2] (v3) at (0,0) {};
\node [circle, fill, inner sep=2] (v4) at (0,-1) {};
\node [circle, fill, inner sep=2] (v5) at (0,-2) {};

\node [circle, fill, inner sep=2] (vw4o) at (1,-0.625) {};
\node [circle, fill, inner sep=2] (vw4u) at (1,-1.375) {};

\node [circle, fill, inner sep=2] (w1) at (2,2) {};
\node [circle, fill, inner sep=2] (w2) at (2,1) {};
\node [circle, fill, inner sep=2] (w4) at (2,-1) {};
\node [circle, fill, inner sep=2] (w5) at (2,-2) {};

\node [circle, fill, inner sep=2] (wr1o) at (3,2.375) {};
\node [circle, fill, inner sep=2] (wr1u) at (3,1.625) {};

\node [circle, fill, inner sep=2] (r1) at (4,2) {};
\node [circle, fill, inner sep=2] (r2) at (4,1) {};
\node [circle, fill, inner sep=2] (r3) at (4,0) {};
\node [circle, fill, inner sep=2] (r4) at (4,-1) {};
\node [circle, fill, inner sep=2] (r5) at (4,-2) {};

\draw [thick, rounded corners] (-4.25,0) -- (-4.25,2.25) -- (-3.75,2.25) -- (-3.75,1.25) -- (-1.75,1.25) -- (-1.75,0.75) -- (-3.75,0.75) -- (-3.75,-1.75) -- (-1.75,-1.75) -- (-1.75,-2.25) -- (-3.75,-2.25) -- (-4.25,-2.25) -- (-4.25,0);
\draw [thick, rounded corners] (4.25,0) -- (4.25,2.25) -- (3.75,2.25) -- (3.75,1.25) -- (1.75,1.25) -- (1.75,0.75) -- (3.75,0.75) -- (3.75,0.25) -- (-0.25,0.25) -- (-0.25,-0.25) -- (3.75,-0.25) -- (3.75,-2.25) -- (4.25,-2.25) -- (4.25,0);

\draw [thick, rounded corners] (-2.25,2.25) rectangle (2.25,1.75);
\draw [thick, bend right=-5] (l1) edge (u1);
\draw [thick, bend right=5] (l1) edge (u1);
\draw [thick, bend right=0] (w1) edge (r1);

\draw [thick, bend right=0] (w1) edge (wr1o);
\draw [thick, bend right=-15] (w1) edge (wr1o);
\draw [thick, bend right=0] (wr1o) edge (r1);
\draw [thick, bend right=-15] (wr1o) edge (r1);
\draw [thick, bend right=0] (w1) edge (wr1u);
\draw [thick, bend right=15] (w1) edge (wr1u);
\draw [thick, bend right=0] (wr1u) edge (r1);
\draw [thick, bend right=15] (wr1u) edge (r1);

\draw [thick, bend right=-5] (u2) edge (v2);
\draw [thick, bend right=5] (u2) edge (v2);
\draw [thick, bend right=-5] (v2) edge (w2);
\draw [thick, bend right=5] (v2) edge (w2);

\draw [thick, bend right=0] (l3) edge (u3);
\draw [thick, bend right=0] (u3) edge (v3);

\draw [thick, bend right=0] (l3) edge (lu3o);
\draw [thick, bend right=-15] (l3) edge (lu3o);
\draw [thick, bend right=0] (lu3o) edge (u3);
\draw [thick, bend right=-15] (lu3o) edge (u3);
\draw [thick, bend right=0] (l3) edge (lu3u);
\draw [thick, bend right=15] (l3) edge (lu3u);
\draw [thick, bend right=0] (lu3u) edge (u3);
\draw [thick, bend right=15] (lu3u) edge (u3);

\draw [thick, bend right=0] (u3) edge (uv3o);
\draw [thick, bend right=-15] (u3) edge (uv3o);
\draw [thick, bend right=0] (uv3o) edge (v3);
\draw [thick, bend right=-15] (uv3o) edge (v3);
\draw [thick, bend right=0] (u3) edge (uv3u);
\draw [thick, bend right=15] (u3) edge (uv3u);
\draw [thick, bend right=0] (uv3u) edge (v3);
\draw [thick, bend right=15] (uv3u) edge (v3);

\draw [thick, rounded corners] (-2.25,-1.25) rectangle (0.25,-0.75);
\draw [thick, bend right=-10] (l4) edge (u4);
\draw [thick, bend right=0] (l4) edge (u4);
\draw [thick, bend right=10] (l4) edge (u4);
\draw [thick, bend right=-5] (v4) edge (w4);
\draw [thick, bend right=5] (v4) edge (w4);
\draw [thick, bend right=-10] (w4) edge (r4);
\draw [thick, bend right=0] (w4) edge (r4);
\draw [thick, bend right=10] (w4) edge (r4);

\draw [thick, bend right=0] (v4) edge (vw4o);
\draw [thick, bend right=-15] (v4) edge (vw4o);
\draw [thick, bend right=0] (vw4o) edge (w4);
\draw [thick, bend right=-15] (vw4o) edge (w4);
\draw [thick, bend right=0] (v4) edge (vw4u);
\draw [thick, bend right=15] (v4) edge (vw4u);
\draw [thick, bend right=0] (vw4u) edge (w4);
\draw [thick, bend right=15] (vw4u) edge (w4);

\draw [thick, bend right=-10] (u5) edge (v5);
\draw [thick, bend right=0] (u5) edge (v5);
\draw [thick, bend right=10] (u5) edge (v5);
\draw [thick, bend right=-10] (v5) edge (w5);
\draw [thick, bend right=0] (v5) edge (w5);
\draw [thick, bend right=10] (v5) edge (w5);
\draw [thick, bend right=-10] (w5) edge (r5);
\draw [thick, bend right=0] (w5) edge (r5);
\draw [thick, bend right=10] (w5) edge (r5);

\node at (-4,2.5) {$s$};
\node at (4,2.5) {$t$};
\node at (-3,2.75) {$x_1$};
\node at (-1,2.75) {$x_2$};
\node at (1,2.75) {$x_3$};
\node at (3,2.75) {$x_4$};

\node at (-5.25,2) {$\{ x_1, \bar{x}_4 \}$};
\node at (-5.25,1) {$\{ x_2, x_3 \}$};
\node at (-5.25,0) {$\{ \bar{x}_1, \bar{x}_2 \}$};
\node at (-5.5,-1) {$\{ x_1, \bar{x}_3, x_4 \}$};
\node at (-5.5,-2) {$\{ x_2, x_3, x_4 \}$};
\end{tikzpicture}
\caption{The finished multigraph $G'$ for the same family of clauses as in Figure~\ref{fig_step1}. Again, the vertices in a rectangle are thought of as identified, and the edge conflicts as well as the $st$-edges are not drawn.}
\label{fig_step2}
\end{figure}
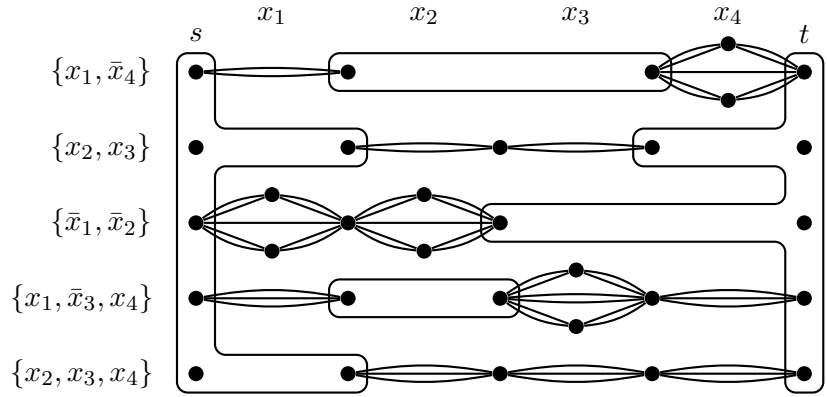

This finishes the construction of $G'$; see Figure~\ref{fig_step2} for an example.
Note that each edge has exactly one edge conflict.
The graph $G'$ has at most $2 + 2m = \calO(n)$ vertices and at most $16n + 2n = \calO(n)$ edges, since every variable occurs exactly three times in $\calF$.
We conclude the construction with a few remarks.
In the later, we will refer to the induced subgraphs obtained by modifying the paths $P_C$ (that is, replacing an edge $uv \in E(P_C)$ by a cycle of length four, wherein $u$ and $v$ are non-adjacent, for variable gadgets) also as $P_C$.
Note that the subgraphs $P_C$ do not comprise copied edges of any form.
Also, we will continue to refer to the copies of the removed edges $u_{i,2}v_{i,2}$ as the copies of $u_{i,2}v_{i,2}$.
To the vertices of the variable gadgets we refer to as in Figure~\ref{fig_gadgets}.


Claim~\ref{claim-st-cut} asserts that any conflict-free cut of $G'$ is an $s$-$t$ cut, and is used later in the proof.
\begin{claim} \label{claim-st-cut}
If there is a conflict-free cut $F$ in $G'$, then it separates $s$ and $t$.
\end{claim}
\renewcommand{\qedsymbol}{$\blacksquare$}
\begin{proof}
Assume for a contradiction that there is a conflict-free cut $F$ of $G'$ that does not separate $s$ and $t$. 
We may assume $F$ to be minimal, that is, $G'-F$ has exactly two components, and $s$ and $t$ are contained in a single one of them.
Let $v$ be a vertex of a component of $G-F$ that does not contain $s$ and $t$, and
let $C$ be the unique clause such that $v \in V(P_C)$.
Note that the set of edges incident to a vertex $a_i$ or $b_i$, respectively, are not conflict-free.
Therefore we may assume that $v$ is incident to copied edges.
The conflicts of the copied edges only admit $P_C$ to be separated into two components by a conflict-free cut. 
Since $F$ is conflict-free, there is an $s$-$v$ or a $v$-$t$ path in $G-F$, a contradiction.
\end{proof}
\renewcommand{\qedsymbol}{$\square$}

We verify that $\calF$ is satisfiable if and only if $G'$ admits a conflict-free cut.
Assume, for one direction, that $\alpha$ is a satisfying assignment of $\calF$. Let
\begin{align*}
F' &= 
\bigcup_{\substack{i \in [n]: \\ x_i \text{ is positive} \\ \text{and }\alpha(x_i)=\texttt{true}}} 
\left( S(u_{i,1}v_{i,1}) \cup S(u_{i,3}v_{i,3}) \right)
\quad\cup \bigcup_{\substack{i \in [n]: \\ x_i \text{ is positive} \\ \text{and }\alpha(x_i)=\texttt{false}}} 
\left( S(a_iv_{i,2}) \cup S(b_iv_{i,2}) \right).
\end{align*}
The multiset of edges $F'$ may not be conflict-free, but we can choose a conflict-free subset in the following way. 
Since $\alpha$ is a satisfying assignment of $\calF$, for every clause $C \in \calF$, there is a variable $x_{i_C}$ that occurs in $C$ and satisfies $C$ under the assignment $\alpha$;
in other words, for every clause $C \in \calF$, there is an index $i_C \in [n]$ such that
\begin{itemize}
    \item there is $j_C \in [3]$ and all edges of the form $u_{i_C,j_C}v_{i_C,j_C}$ are in $F'$, or
    \item all edges of the form $a_{i_C}v_{i_C,2}$ or $b_{i_C}v_{i_C,2}$ are in $F'$.
\end{itemize}
For every $C \in \calF$, we fix exactly one such occurrence, and we let $F_1$ be the subset of $F'$ consisting of these edges.
Moreover, let $F_2$ be the set of $st$-edges in $G'$.
Then $F = F_1 \cup F_2$ is a conflict-free $s$-$t$ cut.

For the other direction, let $F$ be a conflict-free cut of $G'$. 
As we have seen in Claim~\ref{claim-st-cut}, $F$ is actually a conflict-free $s$-$t$ cut of $G'$. 
If $F_2$ denotes the set of $st$-edges in $G'$ again, then we must have $F_2 \subseteq F$.
Together with the fact that $F$ is conflict-free, this implies, if there is only one edge $a_iv_{i,2}$ or $b_iv_{i,2}$ in $F$, then there still is an $u_{i,2}$-$v_{i,2}$ path over $a_i$ or $b_i$ in $G'-F$. 
Therefore we may safely remove such an edge from $F$, and add its conflicting edge $u_{i,j}v_{i,j}$, $j \in \{1,3\}$, to $F$. 
In particular, we may assume without loss of generality that for every 
variable $x_i$ either $u_{i,1}v_{i,1}$ and $u_{i,3}v_{i,3}$ is in $F$, or $a_iv_{i,2}$ and $b_iv_{i,2}$ is in $F$. 
We define a satisfying assignment $\alpha$ of $\calF$.
For all $i \in [n]$, we define
\begin{align*}
    \alpha(x_i) &= \begin{cases}
    \texttt{true}, &\text{if $u_{i,1}v_{i,1}, u_{i,3}v_{i,3} \in F$} \\
    \texttt{false}, &\text{if $a_iv_{i,2}, b_iv_{i,2} \in F$}
    \end{cases}.
\end{align*}
With our assumptions, the assignment $\alpha$ is well-defined, and since $F$ is an $s$-$t$ cut, it satisfies $\calF$.

To complete the proof, we show how to construct a simple graph $G$ with edge conflicts such that $G$ has a conflict-free cut if and only if $G'$ has a conflict-free cut. 
To obtain $G$ from $G'$, we simply replace every vertex of $G'$ with a sufficiently large uncutable gadget, and choose different ends in the uncutable gadgets for every multiedge incident to the replaced vertex in $G'$.
The idea is that each uncutable gadgets must completely be completely contained in one component after removing a conflict-free cut, it therefore mimics the behaviour of a vertex of $G$.
Let $v \in V(G)$.
If $v \neq s$ and $v \neq t$, then $d_{G'}(v) = \calO(1)$, and we only need an uncutable gadget of order $\calO(1)$ to replace $v$.
If $v = s$ or $v = t$, then $d_{G'}(v) = \calO(n)$, since there are $\calO(n)$ many $st$-edges in $G'$ and $m = \calO(n)$ clauses in $\calF$.
Thus, the constructed graph $G$ has $\calO(n)$ vertices and $\calO(n)$ edges.
Since all vertices of an uncutable gadget have degree four except two with degree five, the graph $G$ can be constructed to have maximum degree five.
Therefore $G$ has the desired properties.
The construction can clearly be done in polynomial time, which completes the proof.
\end{proof}

The proof of Theorem~\ref{thm_np} implies a lower bound based on the Exponential Time Hypothesis (ETH) of Impagliazzo and Paturi~\cite{impagliazzo2001complexity}.
We omit stating its formal definition, since we rely on the following result of Cygan, Marx, Pilipczuk and Pilipczuk~\cite{cygan2017hitting}, which in turn is based on the Sparsification Lemma~\cite{impagliazzo2001problems}.
\begin{lemma}[Lemma 6 in \cite{cygan2017hitting}] \label{lem_eth_cleansat}
Unless ETH fails, there does not exist an algorithm that can resolve satisfiability of a clean $n$-variable $m$-clause 3-CNF formula in time $2^{o(n+m)}$.
\end{lemma}
The following corollary is a direct consequence of Lemma~\ref{lem_eth_cleansat} and the proof of Theorem~\ref{thm_np}.

\begin{corollary}
Unless ETH fails, there does not exist an algorithm that decides \cfc{} in time $2^{o(n)}$, where $n$ is the number of vertices of the input graph $G$. This holds, even when $G$ has maximum degree at most five and $\hat{G}$ is 1-regular.
\end{corollary}


\section{Parameterized Complexity Results} \label{sec_param}
\cfc{} is expressible in the (extended) monadic second-order logic (MSO$_2$) using a structure called \emph{hypergraph representation}.
If an instance $(G, \hat{G})$ is expressed in this way, then the \emph{primal graph} $H(G, \hat{G}) = (V, E)$ of this structure is the graph with
\begin{itemize}
    \item vertex set $V = V(G) \cup E(G)$, and
    \item edge set $E = \{ ve: v \in V(G), e \in E(G), v \in e \} \cup E(\hat{G})$.
\end{itemize}
With Courcelle's Theorem~\cite{courcelle1990graph} as stated in Flum and Grohe's book Parameterized Complexity Theory~\cite{flum2006parameterized} in Chapter~11 (where the reader can find one more on this topic), we get the following.
\begin{proposition}
\cfc{} can be decided in linear time on instances $(G, \hat{G})$ for which the primal graph $H(G, \hat{G})$ has bounded treewidth.
\end{proposition}

Next, we show a fixed-parameter tractability result.
\begin{theorem}\label{vertexcover}
\cfc{} is fixed-parameter tractable with the vertex cover number of the input graph $G$ as a parameter.
\end{theorem}
\begin{proof}
Let $G$ be the input graph, let $\hat{G}$ be the conflict graph, and let $X$ be a vertex cover of $G$ of minimum cardinality $k$.
Such a vertex cover can clearly be computed in $2^k\cdot n^{\calO(1)}$-time by the standard bounded search tree algorithm.
In the following, we describe a fixed-parameter algorithm for \cfc{}.
We guess a disjoint partition $(A, B)$ of $X$.
Then we exhaustively apply the following rules:
\begin{decisionrule} \label{drule_vc1}
If $E(A, B)$ contains conflicts, the partition $(A, B)$ cannot be extended to induce a conflict-free cut of $G$.
\end{decisionrule}
\begin{reductionrule} \label{rrule_vc1}
Let $I = V(G) \setminus X$. If there is a vertex $v \in I$ such that there are two incident conflicting edges $a_1v, a_2v \in E(A, I)$ (or $b_1v, b_2v \in E(B, I)$), then add $v$ to $A$ (or add $v$ to $B$).
\end{reductionrule}
\begin{reductionrule}\label{rrule_vc2}
If there is an edge $ab \in E(A, B)$ that has a conflict with an edge $a'v \in E(A, I)$ (or $b'v \in E(B, I)$), where $a, a' \in A$, $b, b' \in B$ and $v \in I$, then add $v$ to $A$ (or add $v$ to $B$).
\end{reductionrule}
\noindent 
It is easy to see that Decision Rule~\ref{drule_vc1} and Reductions Rule~\ref{rrule_vc1} and \ref{rrule_vc2} are correct.

We assume from here that the instance is reduced and all rules are not applicable anymore.
At this point, we construct a 2-SAT formula $\calF$ to decide if $(A, B)$ can be extended to a partition of $V(G)$ that induces a conflict-free cut of $G$.
Let $I = V(G) \setminus X$.
Note that $I$ is independent in $G$.
For each vertex $v \in I$, we introduce a variable $x_v$.
For each edge $\{ uv, rs \} \in E(\hat{G})$ with $u, r \in A \cup B$, $v, s \in I$ and $v \neq s$, we add the following clause to $\calF$:
\begin{itemize}
    \item If $u, r \in A$, we add $\{ \bar{x}_v, \bar{x}_s \}$ to $\calF$.
    \item If $u \in A$ and $r \in B$, we add $\{ \bar{x}_s, x_v \}$ to $\calF$.
    \item If $u \in B$ and $r \in A$, we add $\{ x_s, \bar{x}_v \}$ to $\calF$.
    \item If $u, r \in B$, we add $\{ x_v, x_s \}$ to $\calF$.
\end{itemize}
Now, if $(A', B')$ is a partition of $V(G)$ with $A \subseteq A'$ and $B \subseteq B'$ that induces a conflict-free cut, then for $v \in I$ we set $\alpha(x_v) = \texttt{true}$ if and only if $v \in B'$ for a satisfying assignment of $\calF$;
a partition of $V(G)$ can be constructed from a satisfying assignment in the opposite way.

Regarding the running time, there are $2^{k-1}$ partitions of $X$.
Decision Rule~\ref{drule_vc1} and Reduction Rules~\ref{rrule_vc1} and \ref{rrule_vc2} can be implemented to run in polynomial time.
The same holds for the construction of $\calF$.
Since $\calF$ is a 2-SAT formula, its satisfiability can be decided in polynomial time~\cite{aspvall1979linear}, too.
Thus, the overall running time of the described algorithm is $2^k \cdot n^{\calO(1)}$, that is, the problem is fixed-parameter tractable with the vertex cover number of $G$ as a parameter.
\end{proof}

Given the fixed-parameterized tractability regarding the vertex cover number of $G$, provided in Theorem~\ref{vertexcover}, one can ask about the parameterized complexity of the problem considering stronger parameters such as the feedback vertex set number of $G$.
Next, we show that determining whether an instance $(G,\hat{G})$ admits a conflict-free cut is \W[1]-hard even when $G$ is a series-parallel graph having a feedback vertex set of size one, and the clique cover number of $\hat{G}$ is the parameter.



\begin{proposition} \label{prop_w1}
Given a graph $G$ and a conflict graph $\hat{G}$ on $E(G)$, it is \textup{\W[1]}-hard to decide whether $(G,\hat{G})$ has a conflict-free cut, even when $G$ is a series-parallel graph having a feedback vertex set of size one, and the clique cover number of $\hat{G}$ is a parameter.
\end{proposition}
\begin{proof}
The proof is by a reduction from \textsc{Multicolored Independent Set}.
In this problem we are given a graph $H$ and a partition $(V_1, \dots, V_k)$ of $V(H)$.
We say a set of vertices $I$ is \emph{colorful} if $|V_i \cap I| \leq 1$ for all $i \in [k]$.
The question is whether $H$ admits a colorful independent set $I$ of cardinality at least $k$. 
We may assume that each set $V_i$ induces a clique in $H$.
From this we construct a graph $G$ with edge conflicts as follows.
\begin{enumerate}
    \item We create two new vertices $s$ and $t$, and add them to $G$.
    \item For $i \in [k]$, we create a new $s$-$t$ path $P_i$ of length $|V_i|$ in $G$. We associate to each edge of $P_i$ exactly one vertex of $V_i$.
    \item We choose the conflict graph isomorphic to $H$, using the correspondence of step 2.
\end{enumerate}

Note that $G$ is a series-parallel graph having a feedback vertex set of size one, and every conflict-free cut of $G$ is an $s$-$t$ cut and contains exactly one edge from each path $P_i$. 
In particular, every conflict-free cut of $G$ (if any) has cardinality exactly $k$, which is an upper bound for the clique cover number of $\hat{G}$. 
Moreover, $H$ admits a colorful independent set of cardinality at least $k$ if and only if $G$ admits a conflict-free cut of cardinality at most (or at least) $k$.
\end{proof}

Since the constructed graph $G$ in the reduction of our proof of Proposition~\ref{prop_w1} has a feedback vertex set of cardinality one, we get the following. 

\begin{corollary}
\cfc{} is para-\NP-hard considering the cardinality of a minimum feedback vertex set of the input graph $G$ as a parameter.
\end{corollary}

Besides, as the clique cover number of $\hat{G}$ upper bounds the independence number of $\hat{G}$, it also bounds the size of any conflict-free cut (if any) of $(G,\hat{G})$.
This directly implies the following.
\begin{corollary}
\cfc{} is \textup{\W[1]}-hard taking the cut size as a parameter, while the \XP-solvability is trivial.
\end{corollary}

\section{Simple Polynomial-Time Solvable Cases and Open Problems} \label{sec_poly_open}
In this section we consider an input graph $G$ with $n$ vertices and $m$ edges together with a 1-regular conflict graph $\hat{G}$ on $E(G)$.
A first polynomial-time solvable case is derived from a simple counting argument.
\begin{proposition} \label{prop_counting}
If $m < 2n$, there exists a vertex $v$ of $G$ such that $\partial_Gv$ is conflict-free.
\end{proposition}
\begin{proof}
If for every vertex $v$ of $G$ it holds that $\partial_Gv$ has a conflict, since we are dealing with a simple graph, there are at least $2n$ edges in $G$.
\end{proof}

The bound of Proposition~\ref{prop_counting} is actually sharp:
Recall that the uncutable gadget is the square of an even cycle on at least eight vertices with two additional edges (see Figure~\ref{fig_gadgets}).
Thus, for every $\epsilon > 0$, there exists a graph $G$ with $n$ vertices and $m$ edges, and a 1-regular conflict graph on $E(G)$, such that $G$ does not possess a conflict-free cut and $m \leq (2+\epsilon)n$.
If we take the square of an odd cycle instead, with a similar edge-conflict structure as in the uncutable gadget and without additional edges, then this graph has $n$ vertices and $2n$ edges, and does not possess a conflict-free cut.
These observations, together with Proposition~\ref{prop_counting}, give rise to the following open problem.
\begin{problemenv}
Is the square of an odd cycle with edge-conflicts as in the uncutable gadget the only 4-regular graph without a conflict-free cut? Is it also the only one when we just impose $m = 2n$?
\end{problemenv}

For the next open problem, note that the two additional edges, which we added to the uncutable gadget during the construction in Section~\ref{sec_np} at last, make the graph non-planar. We were unable to construct a planar graph $G$ with a 1-regular conflict graph $\hat{G}$ that does not have a conflict-free cut. Thus, we have a further open problem.
\begin{problemenv}
If $G$ is planar, is there a conflict-free cut?
\end{problemenv}

By Proposition~\ref{prop_counting}, if the input graph $G$ is 2-degenerate, then there exists a $v$ of $G$ such that $\partial_Gv$ is conflict-free, since $m \leq 2n-3$ holds in this case.
It is an easy exercise to prove this directly using a linear ordering.
Therefore, we may ask this question:
\begin{problemenv}
What is the complexity of \cfc{} on 3-degenerate graphs $G$ and 1-regular conflict graphs $\hat{G}$?
\end{problemenv}


\section{An \NP-Complete Symmetric Variant of SAT} \label{sec_sat}
With Theorem~\ref{thm_np} we can show that a symmetric variant of \textsc{SAT} is \NP{}-complete.
We say that a family of clauses $\calF$ is \emph{symmetric}, if for every clause $C = \{\ell_1, \dots, \ell_k\} \in \calF$ with literals $\ell_1, \dots, \ell_k$ the \emph{symmetric clause} $\bar{C} = \{\bar{\ell}_k, \dots, \bar{\ell}_k\}$ of $C$ is in $\calF$, too.  
We will consider the following SAT variant:

\begin{framed}
\noindent
\symnontrivialsat{} \\
Instance: A symmetric family of clauses $\calF$, where every clause $C \in \calF$ has
\begin{enumerate}[(i)]
    \item size three and contains at least one positive and one negative variable, or
    \item size four and contains two positive and two negative variables. \label{prob_clause2}
\end{enumerate} 
Question: Does $\calF$ have a non-constant satisfying assignment? That is, does $\calF$ have a satisfying assignment that does not assign the same value to every variable?
\end{framed}

\begin{corollary} \label{cor_np_sym_sat}
\symnontrivialsat{} is \NP{}-complete, even when the following is true. For every two distinct variables $x$ and $y$ and every two distinct clauses $C$ and $C'$, the memberships $x, \bar{y} \in C$ and $\bar{x}, y \in C'$ imply that $C'$ is the symmetric clause of $C$.
\end{corollary}
\begin{proof}
It is obvious that \symnontrivialsat{} is in \NP{}. 
We reduce \cfc{}, where the conflict graph is 1-regular, to \symnontrivialsat{}. 
This problem is \NP{}-complete by Theorem~\ref{thm_np}. 
Let $G$ be a connected graph and let $\hat{G}$ be a 1-regular conflict graph, for which we intend to find a conflict-free cut.
We construct a family of symmetric clauses $\calF$ over the variables $\{x_v \,|\, v \in V(G)\}$.
For every pair of conflicting edges $uv$ and $rs$ of $G$ we add the following clauses to $\calF$:
\begin{itemize}
\item If $uv$ and $rs$ are incident, let $u = r$ and add $\{x_u, \bar{x}_v, \bar{x}_s\}$ as well as $\{\bar{x}_u, x_v, x_s\}$ to $\calF$. 
\item Otherwise we add $\{x_u, \bar{x}_v, x_r, \bar{x}_s\}$, $\{x_u, \bar{x}_v, \bar{x}_r, x_s\}$, $\{\bar{x}_u, x_v, x_r, \bar{x}_s\}$ and $\{\bar{x}_u, x_v, \bar{x}_r, x_s\}$ to $\calF$.
\end{itemize}
The family of clauses $\calF$ is symmetric, has $n(G)$ variables and at most $2m(G)$ clauses. The construction can clearly be done in polynomial time. It is easy to verify that $G$ has a conflict-free cut if and only if $\calF$ has a non-constant satisfying assignment.
\end{proof}

%

\begin{proposition}\label{prop_symsat}
Deciding an instance $\calF$ of \symnontrivialsat{} can be reduced to deciding $n-1$ instances of \textsc{4-SAT}, where $n$ is the number of variables occurring in $\calF$.
\end{proposition}
\begin{proof}
Fix a variable $x$ of $\calF$. 
For every other variable $y$ of $\calF$, $y \neq x$, let $\calF_y$ denote the formula obtained from $\calF$ by assigning to $x$ the value \texttt{false} and to $y$ the value \texttt{true}, and simplifying it afterwards. Now $\calF$ is a yes-instance if and only if at least one of the formulas $\calF_y$ is satisfiable.
\end{proof}

Schöning's probabilistic algorithm for $k$-SAT~\cite{schoning1999probabilistic} has a running time of $\calO^*\left((2-\frac{2}{k})^n\right)$, where $n$ is the number of variables in the input $k$-SAT formula.
When the input formula is not satisfiable, the algorithm always gives the correct answer.
When the input formula is satisfiable, the probability that it returns a wrong answer is bounded by a constant.
Thus, its one-sided error probability can be made exponentially small by executing it a constant number of times.
Moreover, it has been fully derandomized by Moser and Scheder~\cite{moser2011full}, who presented a deterministic algorithm for $k$-SAT with a running time of $\calO^*\left((2-\frac{2}{k})^{n + o(n)}\right)$.
Together with these results, Proposition~\ref{prop_symsat} implies the following.

\begin{corollary}\label{cor_fast}
\symnontrivialsat{}, and thus \cfc{}, can be solved
\begin{enumerate}[(i)]
    \item in $\calO^*(1.5^n)$-time by a probabilistic algorithm with arbitrarily small one-sided error probability, or
    \item in $\calO^*(1.5^{n+o(n)})$-time by a deterministic algorithm,
\end{enumerate}
where $n$ is the number of variables, or the order of the input graph, respectively.
\end{corollary}

We conclude this section with a short remark.
Consider an instance of \cfc{} where all conflicting edges are adjacent.
In this case, the corresponding \symnontrivialsat{} instance is actually a 3-SAT formula.
In view of Proposition~\ref{prop_symsat}, deciding it can be reduced to deciding $n-1$ instances of 3-SAT.
Hence, for such an instance of \cfc{}, we can improve the basis of the exponential running time in Corollary~\ref{cor_fast} to $\frac{4}{3}$.

\end{document}